\documentclass[a4paper]{amsart}
\usepackage[utf8]{inputenc}
\usepackage[T1]{fontenc}
\usepackage[english]{babel}
\usepackage{amsmath,amsthm,amssymb,graphicx,subfigure,scrtime,xypic,tikz}
\usetikzlibrary{arrows}
\usepackage[backgroundcolor=white]{todonotes}
\usepackage{mathabx}

\newtheorem{thm}{Theorem}[section]
\newtheorem{cor}[thm]{Corollary}

\newtheorem{lemma}[thm]{Lemma}
\newtheorem{prop}[thm]{Proposition}
\theoremstyle{definition}

\theoremstyle{remark}
\newtheorem{rem}[thm]{Remark}
\newtheorem{problem}{\bf Problem}

\newtheorem{claim}[thm]{Claim}

\numberwithin{equation}{section}

\newcommand{\struc}[1]{\langle #1 \rangle}

\newcommand{\alg}[1]{\mathbf{#1}}

\newcommand{\me}{\mathrm{\bf m}}

\def\tkzscl{0.8} 

\title{Arrow type impossibility theorems over median algebras }
\date{\today, \thistime}
\author{Miguel Couceiro}
\address{LORIA (CNRS - Inria Nancy Grand Est - Université de Lorraine), \'Equipe Orpailleur,
Batiment B, Campus Scientifique,   B.P. 239,  
 F-54506 Vandoeuvre-lès-Nancy, France}\email{miguel.couceiro[at]inria.fr}
 \author{Stephan Foldes}
\address{Tampere University of Technology,
PL 553, 33101 Tampere, Finland} \email{sf[at]tut.fi}
\author{Gerasimos C. Meletiou}
\address{TEI of Epirus, PO Box 110, 47100 Arta, Greece} \email{gmelet[at]teiep.gr }
\begin{document}

\begin{abstract}
We characterize trees as median algebras and semilattices by relaxing conservativeness. Moreover, we describe  median homomorphisms between products of median algebras and
show that Arrow type impossibility theorems for mappings from a product $\alg{A}_1\times  \cdots \times \alg{A}_n$ of median algebras to a median algebra
$\alg{B}$ are possible if and only if $\alg{B}$ is a tree, when thought of as an ordered structure.
\end{abstract}

\maketitle

\section{Introduction}\label{sec:intro}

Median algebras have been investigated by several authors in the context of distributive algebras, semilattices, graphs and hypergraphs 
and in the framework of convex and metric spaces (see, e.g.,  \cite{Avann1961,Bandelt1983,Barthelemy83,Birkhoff1947,Evans82,Monjardet76,Monjardet80,Mulder79,Sholander1954,Slater80,Verheul93}).
A \emph{median algebra} is a structure $\alg{A}=\struc{A, \me}$ for a set $A$ and a ternary symmetric operation $\me\colon A^3\to A$, called \emph{median}, such that
\begin{gather*}
\me(x,x,y)=x,\\
\me(\me(x,y,z), t, u)=\me(x,\me(y,t,u), \me(z,t,u)).
\end{gather*}
In fact, axiom systems using only 4 variables are also known \cite{Bandelt1983,Isbell1980,Kol-Mar74}. For instance, in \cite{Kol-Mar74} it was shown that 
$$\me(x,x,y)=x\quad \text{and} \quad \me(\me(x,w,z),y,z)=\me(\me(y,z,w),x,z)$$
suffice to axiomatize median algebras. In particular, it immediately follows that 
$$\me(\me(x,y,z), y, z)=\me(x,y,z).$$

It is well known \cite{Sholander1954} that each element $a$ of a median algebra $\alg{A}$ gives rise to a \emph{median semilattice} 
$\struc{A, \leq_a}$ where $\leq_a$ is given  by
\[
x \leq_a y \quad \iff \quad \me(a,x,y)=x.
\]
In this semilattice $a$ is the bottom element and the associated binary operation, denoted $\wedge_a$, is  defined by $x\wedge_a y =\me(a,x,y)$.
As shown in \cite{Sholander1954}, median semilattices coincide exactly with those $\wedge$-semilattices whose principal ideals 
$$\downarrow x:=\{y\in A\colon y\leq_a x\}$$
are distributive lattices, and such that for any $a,b,c\in A$, $a\wedge b, b\wedge c, c\wedge a$ have a supremum whenever each pair of these meets is bounded above.
In such a case, we can define a median operation by 
\begin{equation}\label{eqn:intro}
\me_\leq(x,y,z)=(x\wedge y)\vee (x\wedge z)\vee (z\wedge y),
\end{equation}
for every $x,y,z \in A$. Moreover, $\me=\me_{\leq_a}$ for every median algebra $\alg{A}=\struc{A, \me}$, and every $a\in A$ \cite{Avann1961}.
Similarly, every distributive lattice gives rise to a median algebra using \eqref{eqn:intro} and the converse also holds whenever there are $a,b\in A$ such that
$\me(a,x,b)=x$ for every $x \in A$.

Another noteworthy connection is  to \emph{median graphs}, i.e., connected graphs having the property that for any three vertices $a,b,c$, 
there is exactly one vertex $x$ in the intersection of the shortest paths between any two vertices in $\{a,b,c\}$. As shown in  \cite{Avann1961},
the \emph{covering graph} (i.e., undirected {Hasse} diagram) of
every median semilattice whose intervals are finite is a median graph. Conversely, 
every median graph is the covering graph of a median semilattice. For further background see, e.g.,  \cite{Bandelt1983}.

In \cite{CMT} the authors studied ``conservative'' median algebras, i.e.,  median algebras that satisfy $\me(x,y,z)\in \{x,y,z\},$ for all $ x, y, z \in A.$
Apart from the 4-element Boolean algebra, it was shown that conservative median algebras 
can be represented by two lower bounded chains whose least elements are identified (thus, they can always be regarded as chains).  
It is noteworthy to observe that they are exactly those median algebras whose subsets are themselves median (sub)algebras (in the terminology of \cite{BandeltVel99},
those with median stabilization degree equal to 0). Equivalently, they were shown to have the 4-element star as a forbidden subgraph. 
The authors in \cite{CMT} also provided complete descriptions of median-homomorphisms between conservative median algebras (with at least 5 elements).

\begin{thm}[{\cite{CMT}}]\label{thm:1}
Let $\alg{A}$ and $\alg{B}$ be two conservative median algebras (thought of as chains) with at least five elements. 
A map $f:\alg{A}\to \alg{B}$  is a median homomorphism if and only if $f$  is  monotone 
(i.e., order-preserving or order-reversing).
\end{thm}

This result was then lifted, by making use of dual topological spaces of median algebras,  to median-homomorphisms between products of conservative median algebras.
For each positive integer $n$, we set $[n]=\{1, \ldots, n\}$.

\begin{thm}[{\cite{CMT}}]\label{thm:2}
Let $\alg{A}=\alg{C}_1\times  \cdots \times \alg{C}_n$ and  $\alg{B}=\alg{D}_1\times \cdots\times \alg{D}_k$ be two finite products of conservative median algebras (thought of as chains). 
Then  $f:\alg{A}\to \alg{B}$ is a median homomorphism if and only if there exist $\sigma:[k]\to[n]$ and monotone maps 
$f_i:\alg{C}_{\sigma(i)}\to \alg{D}_{(i)}$ for $i\in [k]$,   
such that $f=(f_{\sigma(1)}, \ldots, f_{\sigma(n)})$. 

In the particular case when  $k=1$,  
\begin{equation*}
f:\alg{C}_1\times \cdots \times \alg{C}_n\to \alg{D} 
\end{equation*}
 is a median homomorphism 
if and only if there is $j\in [n]$ and a monotone map $g:\alg{C}_j \to \alg{D}$ such that $f=g\circ \pi_j$.
\end{thm}

In this paper we are interested in median algebras $\alg{A}$ that are trees (i.e., acyclic undirected graphs) when thought of as semilattices:
a $\wedge$-semilattice (resp. $\vee$-semilattice) is said to be a \emph{tree} if no pair of incomparable elements have an upper (resp. lower) bound.  
As we will see, such median algebras are obtained by relaxing 
the condition of being conservative: we say that a median algebra $\alg{A}$ is a $(2:3)$-\emph{median semilattice} if 
$\me(x,y,z)\in \{x\wedge y,y\wedge z,z\wedge x\}$ for all $ x, y, z \in A$.
Following the same structure as in \cite{CMT}, we will then proceed to describing the median homomorphisms between (products of) such median algebras. 
As conservative median algebras with more than 4 elements  are particular cases of $(2:3)$-median semilattices, these descriptions properly extend results in \cite{CMT}. 
Unlike in the latter paper, where an extensive use of dual spaces of median algebras took place, in the current paper we take a purely algebraic approach.

The paper is organized as follows. After recalling some terminology and basic results on median algebras in Section~\ref{Preliminary}, we introduce the notion of 
$(2:3)$-median semilattice and show, in Section~\ref{sec:TreeMed}, that it captures those median semilattices that are trees. The problem of
describing homomorphisms between products of median algebras is then tackled in Section~\ref{sec:HomTree}. We show that such median preserving mappings can be decomposed into
 homomorphisms of the form
 \begin{equation}\label{eq:manymedto1}
  f:\alg{A}_1\times  \cdots \times \alg{A}_n\to \alg{B}, 
 \end{equation}
 and explicitly describe them in the case when $\alg{B}$ is a tree.
  Our results are illustrated by several examples and their limitations by counter-examples.
 In Section~\ref{final} we comment on natural interpretations of our results, namely, when looked at as impossibility theorems. 
 In particular, we show that Arrow type impossibility results for median-preserving proceedures \eqref{eq:manymedto1} 
 (that basically state that they depend on at most one argument) can hold if and only
 if $\alg{B}$ is a tree.

\section{Preliminary results}\label{Preliminary}

In this section we introduce basic notions, terminology and notation, as well as
 recall some basic facts and well-known  results
about median algebras as ordered sets, that will be used throughout the paper. To avoid a lengthy preliminary section, we will restrict ourselves
to strictly necessary background, and we refer the reader to \cite{Birkhoff1948,Davey2002,Gratzer2003} for further background.   
 
As discussed in Section~\ref{sec:intro}, when there is no danger of ambiguity, 
we will not distinguish between median algebras and median semilattices. They will be denoted 
by bold roman capital letters $\alg{A}, \alg{B}, \ldots$, while their universes will be denoted by  italic roman capital letters $A, B, \ldots$.
We will assume that universes of structures are nonempty, possibly infinite, sets. 

Let $\alg{A}=\struc{A, \me}$ be a median algebra and let $a,b\in A$. The \emph{convex hull} of $\{a,b\}$ or the \emph{interval} from $a$ to $b$,  
denoted by $[a,b]$, is defined by 
$$[a,b]:=\{t\in A\colon t=\me (a,t,b)\}.$$
Intuitively, it is the set of all elements of $A$ in the ``shortest'' paths from $a$ to $b$, thus explaining our choice of notation.

\begin{prop}
 Let $\alg{A}=\struc{A, \me}$ be a median algebra and let $a,b\in A$. Then
 $$[a,b]=\{\me(a,t,b)\colon t\in A\}.$$
\end{prop}
%

Furthermore, we also have the following useful result.

\begin{prop}[{\cite{Birkhoff1947}}]\label{prop:2.2}
 Let $\alg{A}=\struc{A, \me}$ be a median algebra, and consider $a,b,c\in A$.
 Then $$[a,b]\cap[b,c]\cap[a,c]=\{\me(a,b,c)\}.$$
\end{prop}

 Let $\alg{A}=\struc{A, \me}$ be a median algebra and let $a,b\in A$. 
 Each convex hull of a 2-element set $\{a,b\}$ can be endowed with a distributive lattice structure. 
 To this purpose, for $a\in A$ recall the binary operation $\wedge_{a}\colon A^2\to A$ defined by: 
$$x\wedge_{a}y:=\me (a,x,y),\quad \text{for }\, (x,y)\in A^2.$$ 
Note that such an operation is idempotent, commutative and associative. 
Furthermore, it is not difficult to see that on $[a,b]$,  $\wedge_{a}$ and $\wedge_{b}$ verify the absortion laws, and that
for every $s\in [a,b]$, we have
$$a\leq_a s\leq_a b\qquad \text{and}\qquad  b\leq_b s\leq_b a.$$
In fact, we have the following well-known result; see, e.g., \cite{Bandelt1983}.

\begin{prop}\label{distributive-Interv}
 Let $\alg{A}=\struc{A, \me}$ be a median algebra. For every $a,b\in A$, 
 $\struc{[a,b],\wedge_{a},\wedge_{b}}$ is a distributive lattice with $a$ and $b$ as the least and greatest elements, respectively. 
\end{prop}
%
%

Moreover, we also have a general description of median-preserving mappings between two median algebras. 
%
%
%
%

\begin{thm}\label{Hom-Median}
Let $\alg{A}$  and  $\alg{B}$ be median algebras, and consider 
a mapping $f \colon A\to B$.
 Then the following assertions are equivalent.
\begin{itemize}

 \item[$(i)$] $f$ is a median-homomorphism.


 \item[$(ii)$] For all $p\in A$, $f \colon \struc{A, \wedge_{p}}\to \struc{B, \wedge_{f(p)}}$ is a semilattice-homomorphism.

 \item[$(iii)$] For all $p\in A$, $f \colon \struc{A, \leq_{p}}\to \struc{B,  \leq_{f(p)}}$ is an order-homomorphism.

 \item[$(iv)$]  For all $a,b\in A$, $f([a,b])\subseteq [f(a),f(b)]$.
 
 
\end{itemize}
\end{thm}

\begin{proof} We prove the following sequence of implications 
$$(i)\implies (ii)\implies (iii)\implies (iv)\implies  (i)$$
thus showing that they are all equivalent.

$(i)\implies (ii)$: Suppose that $f$ is a median-homomorphism and $p\in A$. Then
$$f(a\wedge_{p}b)=f(\me(p,a,b))=\me(f(p),f(a),f(b))=f(a)\wedge_{f(p)}f(b).$$

$(ii)\implies (iii)$: Straightforward.

$(iii)\implies (iv)$: Let $t\in [a,b]$. Then $a\leq_a t\leq_a b$, and thus $f(a)\leq_{f(a)} f(t)\leq_{f(a)} f(b)$ from which it follows that $f(t)\in [f(a),f(b)]$.
In other words, $f([a,b])\subseteq [f(a),f(b)]$.

$(iv)\implies (i)$: Let $a,b,c\in A$. By Proposition~\ref{prop:2.2}, we have 
$$\{\me(a,b,c)\}=[a,b]\cap[b,c]\cap[c,a].$$
By $(iv)$, we then conclude that 
$$f(\me(a,b,c))\in [f(a),f(b)]\cap[f(b),f(c)]\cap[f(c),f(a)]=\{\me(f(a),f(b),f(c))\}.$$
Therefore, $f(\me(a,b,c))=\me(f(a),f(b),f(c)).$
\end{proof}

\begin{cor}\label{cor:new}
 Let $f\colon A\to B$ be a median homomorphism between two median algebras $\alg{A}$  and  $\alg{B}$. Then, for every $a,b\in A$,   
 $f$ is also a lattice homomorphism from $\struc{[a,b],\wedge_{a},\wedge_{b}}$ to $\struc{[f(a),f(b)],\wedge_{f(a)},\wedge_{f(b)}}$.
\end{cor}


\section{Trees as median algebras}
\label{sec:TreeMed}

In this section we focus on median algebras that are trees. 
In the particular case of chains, it was shown that these median algebras are exactly those that are conservative, i.e.,
$$\me(x,y,z)\in \{x,y,z\}, \quad x, y, z \in A.$$

To identify median semilattices that are trees, we propose the following relaxation of conservativeness.
We say that a median algebra, thought of as a median semilattice $\alg{A}=\struc{A, \wedge}$, is a $(2:3)$-\emph{median semilattice} if 
for every $ x, y, z \in A$, we have 
$$\me(x,y,z):=(x\wedge y)\vee (x\wedge z)\vee (z\wedge y)\in \{x\wedge y,y\wedge z,z\wedge x\}.$$ 

We start with a simple yet useful observation.

\begin{rem}\label{convexTreeChain}
 Let $\alg{A}=\struc{A, \me}$ be a median algebra such that $\struc{A, \leq_c}$ is a tree for some (or, equivalently, all) $c\in A$.
 Then, for all $a,b\in A$, the convex hull $[a,b]$  is a chain from $a$ to $b$. 
\end{rem}

\begin{thm}\label{Median-Tree}
Let $\alg{A}=\struc{A, \me}$ be a median algebra. Then the following assertions are equivalent.
\begin{itemize}

 \item[$(i)$] There is $p\in A$ such that $\alg{A}=\struc{A, \leq_p}$ is a $(2:3)$-median semilattice.

 \item[$(ii)$] For all $p\in A$, $\alg{A}=\struc{A, \leq_p}$ is a $(2:3)$-median semilattice.

 \item[$(iii)$] There is $p\in A$ such that  $\alg{A}=\struc{A, \leq_p}$ is a tree.

 \item[$(iv)$]  For all $p\in A$,   $\alg{A}=\struc{A, \leq_p}$ is a tree.
 
 \item[$(v)$]  For all $a,b\in A$,  the bounded distributive lattice $\struc{[a,b],\wedge_{a},\wedge_{b}}$ is a chain.
\end{itemize}
\end{thm}

\begin{proof} We prove the following sequence of implications 
$$(ii)\implies (i)\implies (iii)\implies (iv)\implies (v)\implies (iv)\implies (ii),$$
thus showing that they are all equivalent.

 $(ii)\implies (i)$: Straightforward.

 $(i)\implies (iii)$:  Suppose that $\alg{A}=\struc{A, \leq_c}$ is a $(2:3)$-median semilattice.
 To show that it is a tree it suffices to show that for every pair $a,b\in A$ with a common upper bound,
 we have $a \leq_p b$ or $b \leq_p a$. 
 
 So suppose that $c\in A$ is a common upper bound of $a$ and $b$, that is, $\me(p,a,c)=a$ and $\me(p,b,c)=b$. 
 Since $\alg{A}=\struc{A, \leq_c}$ is a $(2:3)$-median semilattice and $c$ is an upper bound of $a$ and $b$, 
 \begin{equation}\label{eq:1}
  \me(c,a,b)\in \{\me(p,a,b), a, b\}.
 \end{equation}
Note that since $\me$ is a median,   
$$\me(p,\me(a,c,b),a)=\me(b,\me(a,a,p),\me(c,a,p))=\me(b,a,a)=a,$$
and thus $a\leq_p \me(c,a,b)$. Similarly, we have $b\leq_p \me(c,a,b)$.

Now, if $ \me(c,a,b)=\me(p,a,b)$ in \eqref{eq:1}, then $\me(c,a,b)\leq_pa$ since $\me(p,a,b)\leq_pa$. Hence, $a=\me(c,a,b)$.
Similarly, we also have $b=\me(c,a,b)$, and thus $a=b$.  

If $\me(c,a,b)=a$ in \eqref{eq:1}, then 
$$\me(p,a,b)=\me(\me(c,a,b),p,b)=\me(a,\me(c,p,b),\me(b,p,b))=\me(a,b,b)=b,$$
and thus  $b \leq_p a$. 
Similarly, if $\me(c,a,b)=b$ in \eqref{eq:1}, we conclude that $a \leq_p b$.

Since in all possible cases we have that $a \leq_p b$ or $b \leq_p a$,  $\alg{A}=\struc{A, \leq_p}$ is a tree.

 $(iii)\implies (iv)$: Straightforward.
 
 $(iv)\implies (v)$: Let $s,t\in [a,b]$. Thus $a\leq_a s,t\leq_a b$, and $s,t$ have a common upper bound in $\struc{A, \leq_a}$.  As $\struc{A, \leq_a}$ is a tree, $s$ and $t$ 
 cannot be incomparable, and hence the interval $[a,b]$ is a chain.

 $(v)\implies (iv)$: Let $a,b,p\in A$ and suppose that $a$ and $b$ are incomparable  in $\alg{A}=\struc{A, \leq_p}$. Note that $p\leq_p a,b$.
 
 For the sake of a contradiction, suppose that  $a$ and $b$  have a common upper bound $d\in A$, that is,
 $p\leq_p a,b\leq_pd$. However, by $(v)$ the interval $[p,d]$ is a chain and thus we must have $a \leq_p b$ or $b\leq_p a$, which constitutes the desired contradiction.  
 
 $(iv)\implies (ii)$: Let $p\in A$. As $\alg{A}=\struc{A, \leq_p}$ is a tree, for every $a,b,c\in A$, we have that $a\wedge b, b\wedge c$ and $ c\wedge a$ 
 are pairwise comparable, and thus 
 $$\me(a,b,c)\in \{a\wedge b,b\wedge c,c\wedge a\}.$$
 Hence, $\alg{A}=\struc{A, \leq_p}$ is a $(2:3)$-median semilattice.
\end{proof}

\begin{rem}\label{TreeChain} Other equivalent descriptions of trees as median semilattices are given in \cite{Bandelt1983,Sholander1954}.
\end{rem}

\begin{rem} As mentioned, conservative median semilattices $\mathbf{A}$ are $(2:3)$-median semilattices, whenever $|A|\geq 5$.
Hence, Lemma 3 and  Theorem 3 in \cite{CMT} follow as corollaries of Theorem \ref{Median-Tree}.
In fact, the conservative median semilattice $\{0,1\}^2$ is the only ``pathological'' case.
\end{rem}

%

\section{Median-homomorphisms over trees}\label{sec:HomTree}

Let $\alg{A}_1,  \cdots , \alg{A}_n$  and $\alg{B}_1,  \cdots , \alg{B}_k$ be median algebras. We seek to decsribe 
median-homomorphisms of the form
\begin{equation}\label{eq:med-hom}
f\colon \alg{A}_1\times  \cdots \times \alg{A}_n\to \alg{B}_1\times   \cdots \times  \alg{B}_k.
\end{equation}

For each $j\in [k]$, let $\pi_j$ denote the $j$-th projection on $\alg{B}_1\times   \cdots \times  \alg{B}_k$, i.e., the mapping
$$\pi_j\colon \alg{B}_1\times   \cdots \times  \alg{B}_k\to \alg{B}_j.$$
It is easy to see that if $f$ is a median-homomorphism, then 
the composition 
$$g_j=\pi_j\circ f\colon \alg{A}_1\times  \cdots \times \alg{A}_n\to \alg{B}_j$$ is a median-homomorphism. 
As the converse holds trivially, possibly with repeated $g_j$'s (as they are not necessarily pairwise distinct), we get the following result.

\begin{lemma}\label{lem:dec} Let $\alg{A}_1,  \cdots , \alg{A}_n$  and $\alg{B}_1,  \cdots , \alg{B}_k$ be median algebras. 
 A mapping $f\colon \alg{A}_1\times  \cdots \times \alg{A}_n\to \alg{B}_1\times   \cdots \times  \alg{B}_k$ 
 is a median-homomorphism if and only if there are  median-homomorphisms 
 $$g_j\colon \alg{A}_1\times  \cdots \times \alg{A}_n\to \alg{B}_j, \quad j\in [k],$$ such that
$f=(g_{1},\ldots, g_{k})$. Moreover, we have $g_j=\pi_j\circ f$ for the projection 
$$\pi_j\colon \alg{B}_1\times   \cdots \times  \alg{B}_k\to \alg{B}_j,\quad j\in [k].$$
\end{lemma}

Hence the description of median-homomorphisms \eqref{eq:med-hom} boils down to describing  median-homorphisms of the form
$$ f:\alg{A}_1\times  \cdots \times \alg{A}_n\to \alg{B}$$
from a finite product of median algebras $\alg{A}_1,  \cdots , \alg{A}_n$ (thought of as median semilattices)
to a median algebra $\alg{B}$. The general answer to this question still eludes us, but 
we can provide explicit descriptions of such mappings when  $\alg{B}$ is of a special type.
We start by making some useful observations.

A function $ f:\alg{A}_1\times  \cdots \times \alg{A}_n\to \alg{B}$ 
is called an \emph{$n$-median-homomrphism} if its ``unary sections'' $f_i\colon \alg{A}_i\to \alg{B}$ (obtained from $f$ by fixing all but the $i$-th argument)  are median-homomorphisms.
It is easy to verify that every median-homomorphism is an $n$-median-homomorphism.
However, the converse is not true. To see this, let $\alg{A}_1= \alg{A}_2= \alg{B}$ be the two element median algebra, and consider
$f\colon \alg{A}_1\times \alg{A}_2\to \alg{B}$ given by $f(1,1)=f(0,1)=f(1,0)=1$ and $f(0,0)=0$. 
Clearly, $f$ is a 2-median homomorphism but it is not a median-homomorphism: for $\mathbf{a}=(0,0), \mathbf{b}=(0,1)$ and $\mathbf{c}=(1,0)$ we have 
$$
f(\me(\mathbf{a},\mathbf{b},\mathbf{c}))=f(0,0)\neq \me(1,1,0)=\me(f(0,1),f(1,0),f(0,0)).
$$

Now we consider some particular cases, namely, when $\alg{B}$  is a chain (or, equivalently, a conservative median algebra) and 
when it is a $(2:3)$-semilattice (or, equivalently, when $\alg{B}$ is a tree). 

We start with the case when $\alg{B}$ (thought of as median semilattice) is a chain. 

\begin{prop}\label{med-chain-homo}
 Let $\alg{A}_1,  \cdots , \alg{A}_n$ be median algebras and let $\alg{B}$  be a chain with $|B|\geq 2$.
 A mapping $ f:\alg{A}_1\times  \cdots \times \alg{A}_n\to \alg{B}$ is a median-homomorphism if and only if 
 there is an $i\in [n]$ and a median-homomorphism $g\colon  \alg{A}_i\to \alg{B}$ such that
 $$
 f(x_1,\ldots,x_n)=g(x_i),\quad \text{for all}\quad (x_1,\ldots,x_n)\in {A}_1\times \cdots \times {A}_n.
 $$
 In other words, $f$ is an essentially unary median-homomorphism. 
\end{prop}

\begin{proof}
 
 Clearly, sufficiency holds. To prove necessity we show that $f$ cannot depend on two different arguments. 
 
 For the sake of contradiction, suppose that $f$ depends on at least 2 of its arguments. Without loss of generality, we may assume that they are the first two.
 Hence, there are $a,b\in \alg{A}_1$ and $\mathbf{c}\in \alg{A}_1\times  \cdots \times \alg{A}_n$ such that $f(\mathbf{c}^a_1)>f(\mathbf{c}^b_1)$, where 
 $\mathbf{c}^x_i$ denotes the tuple obtained from $\mathbf{c}$ by setting its $i$-th component to $x$.
 
 \begin{claim}\label{claim1}
  For every $\mathbf{d}\in \alg{A}_1\times  \cdots \times \alg{A}_n$, $f(\mathbf{d}^a_1)>f(\mathbf{d}^b_1)$.
 \end{claim}
\begin{proof}
Just consider the tuples $\mathbf{c}^a_1, \mathbf{c}^b_1,\mathbf{d}^a_1$ and $\mathbf{c}^a_1, \mathbf{c}^b_1,\mathbf{d}^b_1$,
and use the fact that $f$ is a median homomorphism. 
\end{proof}

Similarly, there are  $a',b'\in \alg{A}_2$ and $\mathbf{c}'\in \alg{A}_1\times  \cdots \times \alg{A}_n$ such that $f({\mathbf{c}'}^{a'}_2)>f({\mathbf{c}'}^{b'}_2)$.
Again, we conclude that for every $\mathbf{d}\in \alg{A}_1\times  \cdots \times \alg{A}_n$, we have 
$$f(\mathbf{d}^{a'}_2)>f(\mathbf{d}^{b'}_2).$$
Observe that, in particular, $f(\mathbf{d}^{a,a'}_{1,2})\geq f(\mathbf{d}^{b,a'}_{1,2}), f(\mathbf{d}^{a,b'}_{1,2}) \geq f(\mathbf{d}^{b,b'}_{1,2}) $.
As $\alg{B}$  is a chain, we must have 
$$
f(\mathbf{d}^{b,a'}_{1,2})\geq f(\mathbf{d}^{a,b'}_{1,2}) \quad \text{or}\quad f(\mathbf{d}^{b,a'}_{1,2})\leq f(\mathbf{d}^{a,b'}_{1,2}).
$$
Without loss of generality, suppose that $\geq $ holds.

 \begin{claim}\label{claim2}
 If $f $ is a median homomorphism, then $f(\mathbf{d}^{a,a'}_{1,2})= f(\mathbf{d}^{b,a'}_{1,2})$.
 Similarly, If $f $ is a median homomorphism, then $f(\mathbf{d}^{a,b'}_{1,2}) = f(\mathbf{d}^{b,b'}_{1,2}) $. 
 \end{claim}
\begin{proof}
We prove the first claim; the second follows analogously.
To see that we cannot have $>$, consider the tuples 
$$f(\mathbf{d}^{a,a'}_{1,2})> f(\mathbf{d}^{b,a'}_{1,2})\geq f(\mathbf{d}^{a,b'}_{1,2})$$
to conclude that $f$ cannot then be a median-homomorphism. 
\end{proof}

As Claim \ref{claim1} and  Claim \ref{claim2} contradict one another, the proof of Proposition~\ref{med-chain-homo} is now complete.
\end{proof}

The case when $\alg{B}$  is a $(2:3)$-semilattice (or, equivalently, when $\alg{B}$  is a tree), follows 
from Proposition~\ref{med-chain-homo} by observing that if $a$ and $b$  are distinct elements of $\alg{B}$, then $[a,b]$ can be thought of as a chain (see Remark~\ref{convexTreeChain}).
In this case, we can reason as in the proof of Proposition~\ref{med-chain-homo} to obtain a more general result, namely,
when $\alg{B}$ is not necessarily a chain but a median algebra whose median graph is a tree.

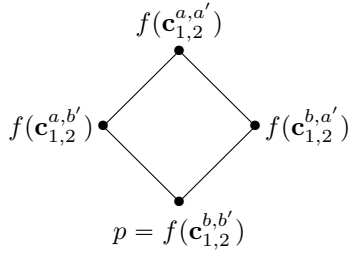
\begin{figure}\label{figure3}
\begin{tikzpicture}
\draw (0,0) node{$\bullet$}node[below]{$p=f(\mathbf{c}^{b,b'}_{1,2})$}--
(-1,1)node{$\bullet$}node[left]{$f(\mathbf{c}^{a,b'}_{1,2})$}--
(0,2)node{$\bullet$}node[above]{$f(\mathbf{c}^{a,a'}_{1,2})$} ;
\draw (0,0)--(1,1)node{$\bullet$}node[right]{$f(\mathbf{c}^{b,a'}_{1,2})$}--(0,2) ;
\end{tikzpicture}
 \caption{A square as in proof of Proposition~\ref{med-tree-homo}.}\label{fig:proof}
\end{figure}

\begin{prop}\label{med-tree-homo}
 Let $\alg{A}_1,  \cdots , \alg{A}_n$ be median algebras and let $\alg{B}$  be a tree.
 A mapping $ f:\alg{A}_1\times  \cdots \times \alg{A}_n\to \alg{B}$ is a median-homomorphism if and only if 
 there is an $i\in [n]$ and a median-homomorphism $g\colon  \alg{A}_i\to \alg{B}$ such that
 $$
 f(x_1,\ldots,x_n)=g(x_i),\qquad \text{for all}\quad (x_1,\ldots,x_n)\in {A}_1\times \cdots \times {A}_n.
 $$
 In other words, $f$ is an essentially unary median-homomorphism. 
\end{prop}

\begin{proof} We provide an additional proof, alternative to that indicated in the paragraph preceding Proposition~\ref{med-tree-homo}.

As in the proof of Proposition \ref{med-chain-homo}, suppose that $f$ depends on at least 2 of its arguments. Without loss of generality, we may assume that they are the first two.
 Hence, there are $a,b\in \alg{A}_1$ and $\mathbf{c}\in \alg{A}_1\times  \cdots \times \alg{A}_n$ such that $f(\mathbf{c}^a_1)\neq f(\mathbf{c}^b_1)$.
 As in Claim \ref{claim1}, it is easy to see that the same holds for all $\mathbf{d}\in \alg{A}_1\times  \cdots \times \alg{A}_n$.
 
 Similarly, if the second argument of $f$ is essential, then are 
 $a',b'\in \alg{A}_2$ and $\mathbf{c'}\in \alg{A}_1\times  \cdots \times \alg{A}_n$ such that $f(\mathbf{c'}^{a'}_2)\neq f(\mathbf{c'}^{b'}_2)$.
 Once again, we have in fact that $f(\mathbf{d'}^{a'}_2)\neq f(\mathbf{d'}^{b'}_2)$, for all $\mathbf{d'}\in \alg{A}_1\times  \cdots \times \alg{A}_n$.

 Hence, we have that  $f(\mathbf{c}^{x,y}_{1,2})$ for $(x,y)\in \{a,b\}\times\{a',b'\}$, are pairwise distinct. By picking $p$ of the form $f(\mathbf{c}^{x,y}_{1,2})$, 
 it then follows that the four points $f(\mathbf{c}^{x,y}_{1,2})\in \alg{B}$ form a square as in Figure~\ref{figure3}. 
 This contradicts the fact that  $\alg{B}$ is a tree.
\end{proof}

By Proposition~\ref{med-tree-homo}  to describe median-homomorphisms of the form 
$$ f:\alg{A}_1\times  \cdots \times \alg{A}_n\to \alg{B}$$ in the case when  
$\alg{A}_1,  \cdots , \alg{A}_n$ are median algebras and $\alg{B}$ is a tree, it suffices to describe median homomorphisms $f\colon\alg{A}\to \alg{B}$ for a median algebra $\alg{A}$ and 
a tree $\alg{B}$. Such descriptions follow from Theorem~\ref{Hom-Median}.

In the case when both $\alg{A}$ and  $\alg{B}$ are trees, Theorem~\ref{Hom-Median} together with Corollary \ref{cor:new} and Theorem \ref{Median-Tree}, imply the following proposition. 

\begin{prop}\label{Hom-Median-Tree}
Suppose that the median algebras $\alg{A}$ and  $\alg{B}$ are trees, and consider 
a mapping $f \colon A\to B$.
 Then the following assertions are equivalent.
\begin{itemize}

 \item[$(i)$] $f$ is a median-homomorphism.
 
 \item[$(ii)$] For all $a,b\in A$, the induced mapping $f \colon \struc{[a,b], \leq_{a}}\to \struc{[f(a),f(b)], \leq_{f(a)}}$ is an isotone function between chains.
 

%
%
%
%
%
\end{itemize}
\end{prop}
%
%
%
%
%
%

\begin{figure}
 

\begin{center}
\subfigure[$\alg{A}_1$]{\label{fig:bcv}\begin{tikzpicture}[scale=\tkzscl]
\draw (0,0) node{$\bullet$}node[below]{$a$}--(0,1)node{$\bullet$}node[right]{$b$};
\draw (0,1)--(-1,2)node{$\bullet$}node[left]{$c$};
\draw (0,1)--(1,2)node{$\bullet$}node[right]{$d$};
\end{tikzpicture}}
\hspace{2cm}
\subfigure[$\alg{A}_2$]{\label{fig:lat03_01} \begin{tikzpicture}[scale=\tkzscl]
\draw (0,2)--(0,2.75) node{$\bullet$}node[right]{$d'$};
\draw (0,0) node{$\bullet$} node[below]{$a$}--(-1,1)node{$\bullet$} node[left]{$b$};
\draw (0,0)--(1,1) node{$\bullet$} node[right]{$c$};
\draw (1,1)--(0,2) node{$\bullet$} node[right]{$d$};
\draw (-1,1)--(0,2);
\end{tikzpicture}}
 \end{center}
 

 \caption{Examples of $\wedge$-semilattices.}\label{fig:lat}
\end{figure}
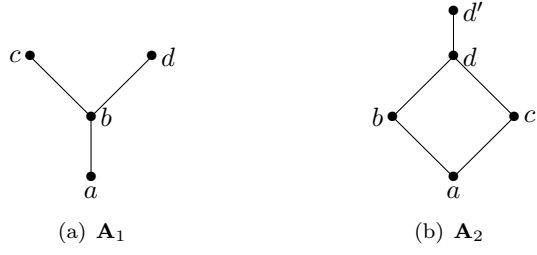

%

In the case when $\alg{A}$ is  a tree and $\alg{B}$ is an arbitrary median algebra, Theorem~\ref{Hom-Median} can be strengthened by an additional and equivalent assertion, namely, 
the existential variant of its assertion $(ii)$. 

\begin{thm}\label{prop:added}
Let $\alg{A}$ be a tree and $\alg{B}$ an arbitrary median algebra. Then we have that one (or, equivalently, all) assertions in  Theorem~\ref{Hom-Median} are equivalent to
\begin{itemize}
\item[$(ii')$] There is a $p\in A$ such that $f \colon \struc{A, \wedge_{p}}\to \struc{B, \wedge_{f(p)}}$ is a semilattice-homomorphism.
\end{itemize}
\end{thm}

\begin{proof}
 We show that $(ii')\implies (i)$. In this way we have 
 $$(ii)\implies (ii')\implies (i)\implies (ii)$$ and thus they are all equivalent.  
 So suppose that $\struc{A,  \wedge_{p}}$ is a tree and that $$f \colon \struc{A, \wedge_p}\to \struc{B, \wedge_{f(p)}}$$ is a semilattice-homomorphism.
 In particular, for every $a,b,c\in A$, we have that
$$\me(a,b,c)\in \{a\wedge_p b,b\wedge_p c,c\wedge_p a\},$$
and that $f \colon \struc{A, \leq_{p}}\to \struc{B,  \leq_{f(p)}}$ is an order-homomorphism. 

Without loss of generality, assume that $\me(a,b,c)=a\wedge_p b$. Hence,
$$a\wedge_p b=\me(a,b,c)=(a\wedge_p b)\vee_p(b\wedge_p c)\vee_p(c\wedge_p a)\geq_p b\wedge_p c,c\wedge_p a,$$
and, from the assumption that $f$ is  a semilattice-homomorphism, 
\begin{eqnarray*}
f(\me(a,b,c))=f(a\wedge_p b)=f(a)\wedge_{f(p)}f( b)\geq_{f(p)} f(b)\wedge_{f(p)}f( c),f(c)\wedge_{f(p)}f( a). 
\end{eqnarray*}
Therefore, 
\begin{eqnarray*}
&\me&(f(a),f(b),f(c))\\
&=&(f(a)\wedge_{f(p)}f( b))\vee_{f(p)} (f(b)\wedge_{f(p)}f( c))\vee_{f(p)}(f(c)\wedge_{f(p)}f( a)) \\
&=& f(a)\wedge_{f(p)}f( b)=  f(\me(a,b,c)) 
\end{eqnarray*}
and the proof of $(ii')\implies (i)$, and thus of the theorem, is now complete.
\end{proof}

\begin{rem}
 Note that the existential variant of $(iii)$ in Theorem~\ref{Hom-Median} is not equivalent to $(i)$ even in the case when $\alg{A}$ is a tree. 
 To see this, consider the function on $\alg{A}_1$ (see Figure~\ref{fig:bcv})
 that maps $a,b$ to $a$, and leaves $c$ and $d$ fixed. Then it is an order-homomorphism for $p=a$, but it is not a median-homorphism.
\end{rem}
%

\begin{rem}
 Note also that Theorem~\ref{prop:added} does not necessarily hold in the case when  $\alg{A}$ is not a tree, even if $\alg{B}$ is conservative.
For instance, let $\alg{A}$ be the median algebra given in Figure~\ref{fig:lat03_01} and  
 $\alg{B}$ the 4-element chain $1,2,3,4$, and consider the mapping 
 that sends $d,d'$ to $4$, $b,c$ to $2$, and $a$ to $1$. Then, for all $p\in A$, 
 $f \colon \struc{A, \wedge_{p}}\to \struc{B, \wedge_{f(p)}}$ is an order-homomorphism, but it is not a semilattice-homorphism for $p=a$ 
 (although it is a semilattice-homorphism for $p=d'$).
\end{rem}
 
\begin{problem} Given an arbitrary median algebra $\alg{B}$, describe those median algebras  $\alg{A}$ for which  Theorem~\ref{prop:added} holds.
\end{problem}



\section{Concluding remarks}\label{final}
 
In this paper we proposed 
a natural relaxation 
of conservativeness as considered in \cite{CMT}, which is of quite different flavour than that proposed in \cite{Bandelt1983,Sholander1954},
and showed that it axiomatizes the class of median algebras that are trees.
Moreover, we considered mappings 
$$ f:\alg{A}_1\times  \cdots \times \alg{A}_n\to \alg{B}_1\times\cdots \times \alg{B}_k$$
from products of arbitrary median algebras to products of $(2:3)$-semilattices (or, equivalently, trees)
and showed that those that preserve the median structure of  $\alg{A}_1\times  \cdots \times \alg{A}_n$ can be decomposed into 
median-homomorphisms $g\colon \alg{A}_i\to \alg{B}_j$. The latter are described in Theorem~\ref{Hom-Median} and, 
in the particular case when $\alg{A}_i$ is a tree, also in Proposition~\ref{Hom-Median-Tree} and in Theorem~\ref{prop:added}.

In the way, we looked into mappings 
$$ f:\alg{A}_1\times  \cdots \times \alg{A}_n\to \alg{B},$$ where $\alg{A}_1,  \cdots , \alg{A}_n$ are arbitrary but where $\alg{B}$ is a $(2:3)$-semilattice, that is,  $\alg{B}$ is a tree. 
The description of such median-homomorphims was then given in Proposition~\ref{med-tree-homo},  from which it follows that they are essentially unary.

Looking at them as aggregation proceedures that are not dictatorial (i.e., that depend on at least two arguments), 
this translates into an impossibility theorem variant to that of Arrow \cite{Arrow50}.

Now the natural question is  to determine whether Proposition~\ref{med-tree-homo} still holds for an arbitrary median algebra  $\alg{B}$. 
More precisely: 

\begin{problem}\label{prob:last}
Given arbitrary median algebras  $\alg{A}_1,  \cdots , \alg{A}_n$, 
 describe those median algebras  $\alg{B}$ for which all median-homomorphisms $ f:\alg{A}_1\times  \cdots \times \alg{A}_n\to \alg{B}$ are trivial, 
 i.e., essentially unary.
\end{problem}

Towards a solution to this problem, consider the following example.
Suppose that  $\alg{B}$ is an arbitrary median algebra thought as a $\wedge$-semilattice that is not a tree.
Thus, there are $a,b\in B$ with a nontrivial upper bound $c$, i.e., $c>a,b>a\wedge b$. Set $p=a\vee b$, $q=a\wedge b$. Note that, together with $a$ and $ b$, they form a square. 
Now,  consider the median algebras $\alg{A}_1= \alg{A}_2=\{0,1\}$. Then $ f:\alg{A}_1\times  \alg{A}_2\to \alg{B}$ given by
$$f(1,1)=p,\quad f(0,1)=a,\quad f(1,0)=b,\quad f(0,0)=q$$
is a median homomorphism that depends on both of its variables. In other words, this is a counter-example to Problem~\ref{prob:last}.

Now this example can be easily extended to mappings $ f:\alg{A}_1\times  \cdots \times \alg{A}_n\to \alg{B}$ and, as we have seen, 
the condition that $\alg{B}$ is a  $(2:3)$-semilattice (i.e., a tree or, equivalently, does not contain a square as an order substructure) forces such 
median-preserving mappings to be essentially unary.
From these considerations, we can thus provide an answer to  Problem~\ref{prob:last}, namely:
\begin{cor}
All median-homomorphisms $ f:\alg{A}_1\times  \cdots \times \alg{A}_n\to \alg{B}$ are essentially unary  if and only if $\alg{B}$ is a tree.
\end{cor}

%

\end{document}